\documentclass[acmsmall, review=false, anonymous=false, screen]{acmart}
\usepackage{csquotes}
\makeatletter
\@ifpackageloaded{cleveref}{}{\usepackage{cleveref}}
\@ifpackageloaded{tabularx}{}{\usepackage{tabularx}}
\@ifpackageloaded{natbib}{\newcommand*{\authorcite}[2][\todo{ADD NAME}]{\citet{#2}}}{\newcommand*{\authorcite}[2][\todo{ADD NAME}]{#1~\cite{#2}}}
\makeatother
\crefformat{footnote}{#2\footnotemark[#1]#3}
\usepackage{booktabs}

\usepackage{transparent}
\usepackage{mathpartir}
\graphicspath{{figures/}}

\usepackage{tikz}
\usepackage{tikzscale}
\usetikzlibrary{backgrounds,calc,positioning}
\usetikzlibrary{graphs}
\usetikzlibrary{graphs.standard}
\usetikzlibrary{arrows.meta}
\usetikzlibrary{arrows}
\usetikzlibrary{shapes}
\usetikzlibrary{fit}
\tikzset{
  REACTIVE/.style={
      circle,
      draw,
      very thick,
      font=\sffamily\tiny,
      label distance=1mm},
  EVENT/.style={REACTIVE, regular polygon,regular polygon sides=4},
  FARROW/.style={thick, arrows={-Triangle}},
  SOURCE/.style={REACTIVE},
  DIST/.style={densely dotted},
  DISTSOURCE/.style={REACTIVE, trapezium, densely dotted},
  MAP/.style={REACTIVE, rectangle},
  FOLD/.style={REACTIVE, regular polygon,regular polygon sides=6},
}

\usepackage{mathtools}

\newcommand{\rulefont}[1]{{\TirNameStyle{#1}}}

\usepackage{color}
\usepackage{listings}

\definecolor{codegreen}{HTML}{008000}
\definecolor{codegray}{HTML}{606060}
\definecolor{codeorange}{HTML}{FFA040}
\definecolor{codered}{HTML}{D04020}
\definecolor{codeblue}{HTML}{06287e}
\definecolor{codelightblue}{HTML}{0e84b5}
\definecolor{codepurple}{HTML}{B00060}
\definecolor{comment}{HTML}{2d2d2d}
\definecolor{light-gray}{gray}{0.92}
\lstdefinestyle{mystyle}{
    backgroundcolor=\color{white},
    commentstyle=\itshape\color{codegreen},
    keywordstyle=\bfseries\color{codeblue},
    keywordstyle= [3]{\bfseries\color{codepurple}},
    keywordstyle= [2]{\bfseries\color{codeblue}},
    keywordstyle= [5]{},
    numberstyle=\color{codegray},
    stringstyle=\color{codeblue},
    basicstyle=\ttfamily\footnotesize,
    frame=tlbr, framesep=0.1cm, framerule=0pt,
    breakatwhitespace=false,
    breaklines=true,
    postbreak=\space\space\space\space,
    captionpos=b,
    keepspaces=true,
    numbers=left,
    stepnumber=1,
    numbersep=10pt,
    showspaces=false,
    showstringspaces=false,
    showtabs=false,
    tabsize=2,
    xleftmargin=20pt,
    xrightmargin=0pt
}
\lstdefinestyle{bw}{
    commentstyle=\itshape\color{gray},
    keywordstyle=\bfseries,
    keywordstyle= [2]{\bfseries},
    keywordstyle= [4]{\bfseries},
    numberstyle=\color{codegray},
    stringstyle=\color{codeblue},
    basicstyle=\ttfamily\footnotesize,
    frame=tlbr, framesep=0.1cm, framerule=0pt,
    breakatwhitespace=false,
    breaklines=true,
    postbreak=\space\space\space\space,
    captionpos=b,
    keepspaces=true,
    numbers=left,
    stepnumber=1,
    numbersep=10pt,
    showspaces=false,
    showstringspaces=false,
    showtabs=false,
    tabsize=2,
    xleftmargin=10pt,
    xrightmargin=10pt
}
\usepackage{sil} %
\usepackage{js}
\usepackage{fr}

\newif\ifshowtodos
\showtodostrue

\newcommand{\langname}{LoRe}

\usepackage[figure, noend, linesnumbered]{algorithm2e}
\SetKw{Each}{each}
\SetKw{Return}{return}
\SetKwProg{Fn}{Function}{:}{end}
\SetKwProg{Proc}{Procedure}{:}{end}
\SetKwData{Reactive}{reactive}\SetKwData{Subgraph}{subgraph}
\DontPrintSemicolon
\definecolor{juliangreen}{HTML}{228b57}

\definecolor{annetteorange}{HTML}{e3782e}

\newtheorem{definition}{Definition}
\newtheorem{theorem}{Theorem}
\newtheorem{lemma}{Lemma}
\acmConference[PL'18]{ACM SIGPLAN Conference on Programming Languages}{January 01--03, 2018}{New York, NY, USA}
\acmYear{2018}
\acmISBN{} %
\acmDOI{} %
\startPage{1}
\setcopyright{none}

\title{\langname: A Programming Model for Verifiably Safe Local-First Software}

\author{Julian Haas}
\orcid{0000-0001-9959-5099}
\affiliation{%
  \institution{Technische Universität Darmstadt}
  \country{Germany}
}

\author{Ragnar Mogk}
\orcid{0000-0003-4583-1791}
\affiliation{%
  \institution{Technische Universität Darmstadt}
  \country{Germany}
}

\author{Elena Yanakieva}
\orcid{0000-0002-2900-7252}
\affiliation{%
  \institution{University of Kaiserslautern-Landau}
  \country{Germany}
}

\author{Annette Bieniusa}
\orcid{0000-0002-1654-6118}
\affiliation{%
  \institution{University of Kaiserslautern-Landau}
  \country{Germany}
}

\author{Mira Mezini}
\orcid{0000-0001-6563-7537}
\affiliation{%
  \institution{Technische Universität Darmstadt}
  \country{Germany}
}

\renewcommand{\shortauthors}{Haas et al.}



\begin{document}

\begin{abstract}
Local-first software manages and processes private data locally while still enabling collaboration between multiple parties connected via partially unreliable networks.
Such software typically involves interactions with users and the execution environment (the outside world).
The unpredictability of such interactions paired with their
decentralized nature make reasoning about the correctness of local-first software a challenging endeavor.
Yet, existing solutions to develop local-first software do not provide support for automated safety guarantees
and instead expect developers to reason about concurrent interactions in an environment 
with unreliable network conditions.

We propose \textit{\langname}, a programming model and compiler that automatically verifies developer-supplied
safety properties for local-first applications.
\textit{\langname} combines the declarative data flow of reactive programming with static analysis and
verification techniques 
to precisely determine concurrent interactions 
that violate safety invariants and to selectively employ strong consistency through coordination where required.
We propose a formalized proof principle and demonstrate how to automate 
the process in a prototype implementation
that outputs verified executable code.
Our evaluation shows that \textit{LoRe} simplifies the development of safe local-first software when compared to state-of-the-art approaches and that verification times are acceptable.
\end{abstract}

\maketitle
\hypertarget{sec:chap:intro}{%
\section{Introduction}\label{sec:chap:intro}}

Applications that enable multiple parties
connected via partially unreliable networks to
collaboratively process data prevail today.
An illustrative example is a distributed calendar application with services to add or modify appointments, where
a user may maintain multiple calendars on different devices,
may share calendars with other users, back them up in a cloud;
calendars must be accessible to users in a variety of scenarios,
including offline periods, e.g., while traveling --
yet, planning appointments may require coordination between multiple parties.
The calendar application is representative for other
collaborative data-driven software such as
group collaboration tools, digital (cross-organizational) supply chains,
multiplayer online gaming, and more.

The dominating software architecture for such applications is centralized:
data is collected, managed, and processed centrally in data centers, while devices on the edge of the communication infrastructure serve primarily as interfaces to users and the outside world. This architecture
simplifies the software running on edge devices since concerns like consistent data changes to ensure
safety properties are managed centrally.
However, this comes with issues including loss of control over data ownership and privacy, insufficient offline availability, poor latency, inefficient use of communication infrastructure, and waste of (powerful) computing resources on the edge.

To address these issues, local-first principles for software development have been formulated~\cite{Kleppmann2019LocalFirst}, calling for moving data management and processing to edge
devices instead of confining the data to clouds.
But for programming approaches that implement these principles to be viable
alternatives to the centralized approach, they must support automatically verifiable
safety guarantees to counter for the simplifying assumptions afforded by a centralized approach.
Unfortunately, existing approaches to programming local-first applications such as \emph{Yjs}~\cite{DBLP:conf/group/NicolaescuJDK16} or \emph{Automerge}~\cite{automerge} do not provide
such guarantees.
They use \emph{conflict-free replicated data types (CRDTs)}~\cite{Shapiro2011}
to store the parts of their state that is shared
across devices and rely on callbacks
for modeling and managing state that changes in both time and space.
CRDTs have been invented in the context of geo-replicated databases
and are available as off-the-shelf
databases~\cite{Shapiro:2018:JRCAntidote} or
libraries~\cite{Kleppmann2017, CRDTOptimizations}.
But the strongest consistency level ensured by CRDTs, causal consistency~\cite{Mahajan2011},
is not enough to maintain invariants that require coordination among the participants,
e.g., the invariant that employees should not enter
more than the available vacation days in a use of the calendar app in a business setting.
At the same time,
we need to delimit the scope of coordination
so as to \enquote{maximize combinations of availability and
consistency that make sense for a particular application}~\cite{Brewer2012}.
To find the set of interactions that actually need coordination, in a local-first setting,
one must reason about possible interleavings of their data flows "end-to-end", i.e.,
from the interface to the outside world, through device-local data-dependency chains,
to remote devices and back.
The unpredictability of the interactions triggered by the outside world,
concurrently at different devices, paired
with the absence of a central authority and the prevailing
implicit dependencies in current callback-centered programming models,
makes such reasoning without automated support a challenging, error-prone endeavour.

To close this gap, we propose a programming model for local-first applications that features explicit safety properties and automatically enforces them.
The model has three core building blocks: \emph{reactives}, \emph{invariants}, and \emph{interactions}.
\emph{Reactives} express values that change in time, but also in space by being replicated over multiple devices.
They enable systematic treatment of
complex state, dependencies, and concurrent changes,
enabling developers to reason in terms of functional composition.
\emph{Invariants} are formula in first-order logic specifying safety properties that must hold at all times
when the application interacts with the outside world, or values of reactives are observable.
\emph{Interactions} interface to the outside world and
encapsulate changes to all reactives affected by interactions with it
(state directly changed by the interactions, device-local
values derived from the changed state, and shared state at remote devices).
They serve as language-managed cross-device data flows that automatically use
best-in-class consistency.
On one device, interactions are logically processed instantaneously, i.e., no
intermediate states are observable, similar to atomicity in databases.
We use automatic verification with invariants as verification obligations to
identify interactions that need coordination across devices, for which
the compiler generates the coordination protocol; all other interactions become visible in causal order.
This way, the compiler makes an application-specific availability-safety trade-off.

The availability-safety trade-off has been explored in the systems and database
community under the term "coordination avoidance"
and there exist approaches that
leverage user-specified safety invariants
to synthesize distributed objects or to correctly configure the consistency level for
each database operation \cite{Balegas2015, Whittaker2018, DBLP:journals/pacmpl/PorreFPB21, Gotsman2016}.
Our work draws inspiration from these approaches, but differs in two key aspects:
(1) Instead of geo-replicated databases,
we target a peer-to-peer local-first setting with unique challenges; specifically, we do not assume any centralized authority and feature offline availability and interactions with the outside world.
(2) Instead of programming applications against the interface to a distributed datastore/object,
which fosters designs split into two separate tiers -- an automatically managed monolithic data store
and the application tier that uses the store's API -- we provide
a language-integrated mixed consistency with whole program guarantees, i.e.,
we verify the safety of derived data “all the way down”.
This is necessary in a local-first setting, where most of the complexity arises from the interactions with the external world
and are handled as part of the application logic, not as part of the data store.

\noindent
In summary, we make the following contributions:
\begin{enumerate}
\item A programming model for local-first applications with verified safety properties (Section~\ref{sec:goals}), called \langname{}.
While individual elements of the model, e.g., CRDTs or reactives, are not novel,
they are repurposed, combined, and extended
in a unique way to systematically address specific needs of local-first applications
with regard to ensuring safety properties.
\item
A formal definition of the model including a formal notion of invariant preservation and confluence for interactions, and a modular verification that invariants are never violated.
In particular, our model enables invariants that reason about the sequential behaviour of the program.
In case of potential invariant violation due to concurrent execution, \langname{} automatically adds the necessary  coordination logic (Section~\ref{sec:programming-model}).
\item A verifying compiler\footnote{\label{repo}The source code of our prototype implementation is available at \url{https://github.com/stg-tud/LoRe}.} that translates \langname{} programs to Viper~\cite{MuellerSchwerhoffSummers16}
for automated verification and to Scala for the application logic including synthesized synchronization to guarantee the specified
safety invariants (Section~\ref{sec:chap:implementation}).
\item An evaluation of \langname{} in two case studies (Section \ref{sec:evaluation}).
Our evaluation validates two claims we make about the programming model proposed, (a)
It facilitates the development of safe
local-first software, and (b) it enables an efficient and modular verification of safety properties.
It further shows that the additional safety properties offered by our model do not come with prohibitive costs in terms of verification effort and time.
\end{enumerate}

\section{\langname{} in a Nutshell}
\label{sec:goals}

We introduce the concepts of \langname{} along the example of a distributed calendar for tracking work meetings and vacation days.
\langname{} is an external DSL that compiles to Scala (for execution) and Viper IR~\cite{MuellerSchwerhoffSummers16} (for verification); its syntax is inspired by both. A \langname{} program defines a distributed application that runs on multiple physical or virtual devices.\footnote{We assume that every device is running the same application code (i.e., the same binary), and different types of devices (such as client and server) are modeled by limiting them to execute a subset of the defined interactions.}
Listing~\ref{lst:calendar} shows a simplified implementation of the calendar example application in \langname{}.
As any \langname{} program, it
consists of replicated state (\lstinline{Source} reactives in Lines~\ref{line:work}-\ref{line:vacation}), local values derived from them (\lstinline{Derived} reactives in Lines~\ref{line:all-appointments}-\ref{line:remaining-vacation}), interactions (Lines~\ref{line:calendar-interaction}-\ref{line:specialize-add-work}), and invariants (Lines~\ref{line:calendar-invariants-one}-\ref{line:calendar-invariants-two}).

\begin{lstlisting}[caption={The distributed calendar application.}, float=t, language=fr, escapechar={§}, label=lst:calendar]
type Calendar = AWSet[Appointment] §\label{line:calendar-awset}§
val work: Source[Calendar] = Source(AWSet()) §\label{line:work}§
val vacation: Source[Calendar] = Source(AWSet()) §\label{line:vacation}§

val all_appointments: Derived[Set[Appointment]]   = Derived{ work.toSet.union(vacation.toSet) }  §\label{line:all-appointments}§
val remaining_vacation: Derived[Int] = Derived{ 30 - sumDays(vacation.toSet) }   §\label{line:remaining-vacation}§

val add_appointment : Unit = Interaction[Calendar][Appointment]  §\label{line:calendar-interaction}§
  .requires{ cal => a => get_start(a) < get_end(a) } §\label{line:calendar-start-end-precondition}§
  .requires{ cal => a => !(a in cal.toSet)}
  .executes{ cal => a => cal.add(a) } §\label{line:calendar-add-executes}§
  .ensures { cal => a => a in cal.toSet } §\label{line:calendar-add-ensures}§
val add_vacation : Unit = add_appointment.modifies(vacation) §\label{line:specialize-interaction}§
  .requires{ cal => a => remaining_vacation - a.days >= 0}
val add_work     : Unit = add_appointment.modifies(work) §\label{line:specialize-add-work}§

UI.display(all_appointments, remaining_vacation) §\label{line:calendar-external-start}§
UI.vacationDialog.onConfirm{a => add_vacation.apply(a)} §\label{line:calendar-callback}§

invariant forall a: Appointment :: §\label{line:calendar-invariants-one}§
  a in all_appointments ==> get_start(a) < get_end(a)

invariant remaining_vacation >= 0 §\label{line:calendar-invariants-two}§
\end{lstlisting}

\begin{figure}
\centering
\begin{tikzpicture}[yscale=3, xscale=4]
\graph[no placement,
        nodes={align=center}] %
{ 
    vacation/"\textbf{vacation}\\\texttt{AWSet[Appointment]}"[x=0, y=0, DISTSOURCE] ->[FARROW]
    remaining/"\textbf{remaining\_vacation}\\\texttt{30 - size(vacation)}"[x=0, y=-0.5, MAP];
    {vacation,
    work/"\textbf{work}\\\texttt{AWSet[Appointment]}"[x=1, y=0, DISTSOURCE]}
    ->[FARROW]
    all/"\textbf{all\_appointments}\\\texttt{union(vacation, work)}"[x=1,y=-0.5, MAP]
};
\end{tikzpicture}
\caption{The data-flow graph of the calendar application.}
\label{fig:calendar-graph}
\end{figure}

\subsection{Reactives}
\label{sec:example-application}

\emph{Reactives} are the composition units in a \langname{} program.
We distinguish two types of them: \emph{source} and \emph{derived} reactives, declared by the keywords \lstinline{Source} and \lstinline{Derived}, respectively. Source reactives are values that are directly changed through interactions.
Their state is modeled as \emph{conflict-free replicated data types} (CRDTs)~\cite{Shapiro2011, preguica2018} and is replicated between the different devices collaborating on the application. Derived reactives represent local values that are automatically computed by the system from the values of other reactives (source or derived). Changes to source reactives automatically (a) trigger updates of derived reactives and (b) cause devices to asynchronously send update messages to the other devices,
which then merge the changes into their local state.
Together, local propagations and asynchronous cross-device update messages ensure
that users always have a consistent view of the overall application state.
All reactives are statically declared in the program source code.
LoRe then statically extracts knowledge about the data flow for modular verification and to minimize the proof goals (cf. Section~\ref{step-2-graph-analysis}).
We discuss the technical implications of static reactives in Section~\ref{sec:chap:conclusion}.

Listing~\ref{lst:calendar} shows two source reactives,
\lstinline{work} and \lstinline{vacation} (Line~\ref{line:work} and~\ref{line:vacation}), each modeling a calendar as a set of appointments.
The work calendar tracks work meetings, while the vacation calendar contains registered vacation days.
When defining a source reactive, programmers have to choose a CRDT for the reactive's internal state.
\langname\ offers a selection of pre-defined CRDTs including various standard data types such as sets, counters, registers and lists.
Further data types can be supported by providing a Viper specification for that data type.
In this case, an \emph{add-wins-set} (a set CRDT where additions have precedence over concurrent deletions) is selected for both
source reactives.
Appointments from both calendars are tracked in the \lstinline{all_appointments} derived reactive (Line~\ref{line:all-appointments}), while the \lstinline{remaining_vacation} reactive (line \ref{line:remaining-vacation})  tracks the number of remaining vacation days.
The \emph{data-flow graph} of the application, where nodes are reactives and edges represent the derivation relation -- in the direction of data flow -- is visualized in Figure~\ref{fig:calendar-graph}. This data-flow graph is created by the \langname{} compiler and managed by its runtime.

\subsection{Interactions}
Changes to the state of the system, e.g., adding appointments to a calendar, happen through explicit \emph{interactions}.
Each interaction has two sets of type parameters:
the types of source reactives that it modifies and the types of parameters that are
provided when the interaction is applied.
For example, the \lstinline{add_appointment} interaction in Line~\ref{line:calendar-interaction} modifies a reactive of type \lstinline{Calendar} and takes a parameter of type \lstinline{Appointment}.
The semantics of an interaction \lstinline{I} are defined in four parts:
(1) \lstinline{requires} (Line~\ref{line:calendar-start-end-precondition}) defines the preconditions that must hold for
\lstinline{I}
to be executed,
(2) \lstinline{executes} (Line~\ref{line:calendar-add-executes}) defines the changes to source reactives,
(3) \lstinline{ensures} (Line~\ref{line:calendar-add-ensures}) defines the postconditions that must hold at the end of \lstinline{I}'s execution,
(4) \lstinline{modifies} (Line~\ref{line:specialize-interaction}) defines the source reactives that \lstinline{I} changes.
The parameters of \lstinline{requires}, \lstinline{executes}, and \lstinline{ensures} are functions that take the modified reactives and the interaction parameters as input (\lstinline{cal} is of type \lstinline{Calendar} and \lstinline{a} is of type \lstinline{Appointment}).
The splitting of the definition of interactions in four parts allows for modularization and reuse.
For instance, \lstinline{add_appointment} is only a partial specification of an interaction, missing the \lstinline{modifies} specification.
Both \lstinline{add_work} (Line~\ref{line:specialize-add-work}) and \lstinline{add_vacation} (Line~\ref{line:specialize-interaction}) specify  complete interactions by adding \lstinline{modifies} to \lstinline{add_appointment}; they are independent interactions that
differ only in their modifies set.

Interactions encapsulate reactions to input from
the outside world (e.g., the callback in Line~\ref{line:calendar-callback} that is triggered by the UI and applies the arguments to \lstinline{add_vacation}).
Applying an interaction checks the preconditions, and -- if they are fulfilled -- computes and applies the
changes to the source reactives, and propagates them to derived reactives --
all in a \enquote{transactional} way in the sense that all changes to affected
reactives become observable at-once (\enquote{atomically}).
Only source reactives are replicated between devices, while derived
reactives are computed by each device individually.
\langname{} gurantees that executing interactions does not invalidate
neither postconditions nor invariants.

\begin{figure*}
\begin{footnotesize}
\ttfamily
\centering
\def\svgwidth{0.7\textwidth}
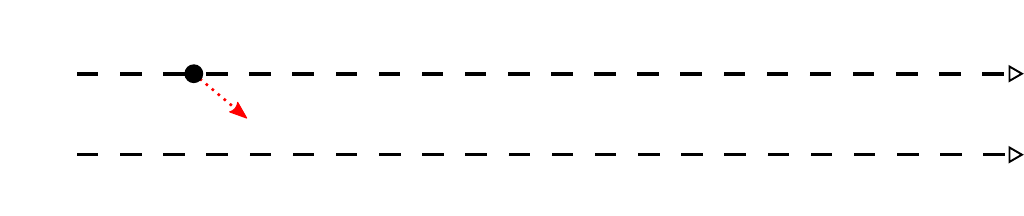
\caption{Concurrent execution of interactions may cause invariant violations.
In this example, device $D_1$ adds a vacation of 20 days to the calendar, while $D_2$ concurrently adds a vacation of 12 days.
Given a total amount of 30 available vacation days, this leads to a negative amount of remaining vacation once the devices synchronize.}
\label{fig:calendar-anomaly}
\end{footnotesize}
\end{figure*}

\subsection{Invariants and Conflicts}
\label{sec:synchronization-points}

\langname{}
 expects the developer to use \emph{invariants}, introduced with the keyword \lstinline{invariant},
to specify application properties that should always hold.
Invariants are first-order logic assertions given to a verifier based on the Viper verification infrastructure~\cite{MuellerSchwerhoffSummers16}.
Invariants can help uncover programming bugs and reveal where the eventually-consistent replication based on CRDTs could lead to safety problems.

For illustration, consider the invariants for the calendar application in Lines~\ref{line:calendar-invariants-one} and \ref{line:calendar-invariants-two}.
The invariant in Line~\ref{line:calendar-invariants-one}
requires that all appointments must start before they end.
Notice, how the invariant can be defined without knowing the amount of calendars and the actual structure of the data-flow graph by simply referring to the \lstinline{all_appointments} reactive.
This invariant represents a form of input validation, and is directly ensured by \lstinline{add_appointment} interactions because the precondition on the arguments requires the added appointment to start before it ends (Line~\ref{line:calendar-start-end-precondition}).
In absence of this precondition, the \langname\ compiler would reject the
program and report a safety error due to a possible invariant violation.
The invariant in Line~\ref{line:calendar-invariants-two} requires that employees do not take more vacation days than available to them.
Again, this is locally enforced by the precondition of the \lstinline{add_vacation} interaction, which ensures that new entries do not exceed the remaining vacation days. But there is nothing stopping two devices from
concurrently adding vacation entries, which in sum violates the invariant.
Figure~\ref{fig:calendar-anomaly} illustrates such a situation:
A user plans a vacation of 20 days on the mobile phone (device \(D_1\)) and later
schedules a 12-day vacation on a desktop (device \(D_2\)), at a time
when $D_1$ was offline. Thus, both interactions happened concurrently
and after merging the states the calendar contains a total of 32
days of vacation, violating the \lstinline{remaining_vacation} invariant.

This example illustrates a conflict between concurrent interactions that can potentially violate an invariant.
In this case, two executions of \lstinline|add_vacation| must not happen concurrently.
A simple remedy for preventing conflicts and the resulting invariant violations is coordinating (synchronizing) the involved interactions.
In the given example, this would mean that two devices could never add vacation entries concurrently, because every execution of \lstinline|add_vacation| would need to be coordinated with the other devices.

\langname{}'s verifying compiler reports conflicting interactions to the developer and automatically generates the required coordination code for the execution of such interactions (see Section~\ref{sec:synchronization}).
The generated coordination code employs a token-based approach~\cite{Gotsman2016} which ensures that devices can only execute conflicting interactions iff they hold the tokens for the interaction itself \emph{and} all its conflicts (in this case, they would only need the \lstinline|add_vacation| token).
Every token can only ever be held by one device at a time, effectively preventing every other device from executing conflicting interactions until they are able to coordinate with the token-holding device and can obtain the required token(s).

This approach ensures program safety but can be at odds with one of the principles behind local-first software: offline-availability.
Usually, adding stronger guarantees will lead to more synchronization, ultimately resulting in an availability/safety trade-off.
In practice, one can improve the situation by employing more sophisticated token distribution algorithms and timeouts\footnote{One example would be assigning a primary device which manages the tokens.
This could mean, e.g., that the user's laptop is always able to perform all interactions (even when offline) while the user's phone needs to be able to contact the laptop for certain conflicting interactions.}
but this still blocks certain devices from performing certain actions at a given time.
LoRe's conflict reporting helps developers to explore
these trade-offs and allows them to make informed decisions about the safety guarantees of their program.
When they find that their program requires too much synchronization,
they can lower the guarantees by adapting their invariants.

\section{Programming Model}
\label{sec:programming-model}

This section formally presents the syntax and semantics of the programming model and discusses how we verify that execution of a program preserves safety guarantees specified by invariants.
The definition of the program execution is split into a big-step semantics for handling reactive updates on a single device and a labelled transition system to model execution of the overall distributed system.
\langname{} guarantees that given a valid program state (that ensures the safety invariants),
any step taken in the labelled transition system preserves the validity of the program.
To preserve safety without sequentializing the whole distributed program execution,
the formal semantics relies on an oracle that tells us when two interactions have a conflict.
We then give a proof of safety preservation given our oracle.
Using the insights of this proof, Section~\ref{sec:chap:implementation} shows
how we use the verification infrastructure to compute this oracle.

\newcommand{\keyword}[1]{\text{#1}}
\newcommand{\transition}[2]{\xRightarrow[#1]{#2}}
\newcommand{\srcset}[0]{\sigma}
\newcommand{\derset}[0]{\delta}
\newcommand{\invset}[0]{\mathit{I}}
\newcommand{\intvar}[0]{a}
\newcommand{\lockvar}{L}
\newcommand{\valid}[1]{\mathit{valid}(#1)}
\newcommand{\preserving}[1]{\mathit{preserving}(#1)}
\newcommand{\conflicts}[1]{\mathit{conflicts}(#1)}
\newcommand{\confluent}[1]{\mathit{confluent}(#1)}
\newcommand{\eval}[3]{#1 \Downarrow_{#2} #3}
\newcommand{\intset}[0]{\mathit{A}}
\newcommand{\transset}[0]{\omega}
\newcommand{\device}[2]{\langle \srcset{}{#1}, L{#2} \rangle}
\newcommand{\bootstrap}[4]{\intset{#1}, \invset{#2}, \srcset{#3}, \derset{#4}}
\newcommand{\devicetrans}[4][l]{\langle\srcset{#2}; \derset{#3}; \transset{#4}; #1\rangle\ }
\newcommand{\invariants}[0]{\textit{Invariants}}
\newcommand{\interactions}[0]{\textit{Interactions}}
\newcommand{\syncpoints}[0]{\textit{Syn}}
\newcommand{\dstep}[2]{(D_1 \mid \ldots \mid D_i \mid \ldots \mid D_n) \xRightarrow[#1]{#2} (D_1 \mid \ldots \mid D_i' \mid \ldots \mid D_n)}

\begin{figure}
\centering
\begin{gather}
\begin{align*}
P ::=         & \ (A, \derset{}, I, (D \mid \ldots \mid D))                     & \mathit{Program}\\
D ::=         & \ \device{}{}                                                   & \mathit{Device} \\
A ::=         & \ \{\keyword{Interaction}((r, \ldots, r), l, l, t)\}            & \mathit{Interactions} \\
\srcset{} ::= & \ \{\keyword{val}\ r = \keyword{Source}(v)\}                    & \mathit{Sources} \\
\derset{} ::= & \ \{\keyword{val}\ r = \keyword{Derived}(t)\}                   & \mathit{Derived} \\
I ::=         & \ \{\keyword{Invariant}(l)\}                                    & \mathit{Invariants} \\
v ::=         & \ r \mid \keyword{true} \mid \keyword{false} \mid x => t        & \textit{Value} \\
t ::=         & \ v \mid t\ t \mid x \mid                                       & \lambda\ \mathit{Term} \\
              & \ r.\keyword{value}                                             & \mathit{Reactive\ Access} \\
l ::=         & \ t \mid x => l \mid \neg l \mid l == l \mid l \wedge l \mid l \vee l \mid & \mathit{Logic\ Term} \\
              & \ \keyword{forall}\ x.\ l \mid \keyword{exists}\ x.\ l          & \textit{Quantifier}
\end{align*}
\end{gather}
\begin{minipage}{0.48\linewidth}
\begin{align*}
x  \qquad                                                &\mathit{Variables}\\
r \qquad                                                &\mathit{Reactive\ Identifiers}\\
L  \subseteq   \intset \qquad                                   &\mathit{Locks}
\end{align*}
\end{minipage}
\caption{Abstract syntax of \langname\ programs.}
\label{fig:lore-syntax}
\label{fig:term-syntax}
\end{figure}

\subsection{Syntax and Evaluation Semantics}

\subsubsection{Syntax}
Figure~\ref{fig:lore-syntax} shows the abstract syntax for \langname{}.
A distributed program $P$ is defined as a tuple, whose elements are
a set of interactions $A$,
a set of derived reactives $\derset{}$, a set of invariants $\invset$, and a list of
devices $(D_1 \mid \dots \mid D_n)$, using $\mid$ as a separator in reference to parallel composition.
We write $D \in P$ to state that $D$ is one of the
devices in $P$,
where a device $D = \device{}{}$ consists of an assignment of source reactives $\srcset{}$ and a set of locks $L$.
For the following definitions, we use curly braces as part of the meta syntax to denote that an expression occurs zero or more times in any order.
This is used for top-level definitions in \langname{},
and we treat such expressions as having set semantics.

Every $\keyword{Interaction}((r_1, \ldots, r_n), l_{pre}, l_{post}, t_{exec})$ is composed of a
list
of affected reactives $(r_1, \ldots, r_n)$, pre- and postconditions $l_{pre}$ and $l_{post}$,
and the interaction body $t_{exec}$.
We introduce an inner term language $t$ for the bodies of reactives and interactions, which is a simple lambda-calculus extended
with access to reactives. We write $v = \srcset{}(r)$ to refer to the current value of a
source reactive, if the expression $\keyword{val}\ r = \keyword{Source}(v)$ is present in $\srcset{}$,
and
$t = \derset{}(r)$ to refer to the body of a derived reactive if the expression $\keyword{val}\ r = \keyword{Derived}(t)$ is present in $\derset{}$.
We use a first-order logic language $l$, which embeds $\lambda$-terms,
for defining invariants, pre-, and postconditions.

$\intset{}$, $\derset{}$, and $\invset{}$ are static and only devices $D = \device{}{}$ may change during execution by updating the values of source reactives $\srcset{}$ or their currently held locks $L$.
Semantically, each interaction $\intvar{} \in \intset{}$ has a one-to-one correspondence to a lock,
thus we syntactically represent locks the same as interactions with $L \subseteq A$.

\subsubsection{Term evaluation}
We use a big-step semantics for term evaluation.
Figure~\ref{fig:term-semantics} shows the rules for reactive evaluation -- we omit rules related to standard lambda-calculus and logic evaluation.
We write $\eval{t}{\srcset{}}{v}$ to state that $t$ evaluates to $v$ given a current assignment of source reactives $\srcset{}$.
Note that evaluation depends on $\derset{}$, but we omit writing it in the rules as it is fixed for each program and thus always clear from context.
Rule \rulefont{ValueSource} retrieves the current value of a source reactive from the store $\srcset{}$.
Rule \rulefont{ValueDerived} evaluates the expression of a derived reactive, which may depend on other reactives,
thus the rule potentially results in evaluating many derived and source reactives in the data-flow graph.
\begin{figure}
\begin{mathpar}
\infer[ValueSource]
{\eval{t}{\srcset{}}{r} \and v = \srcset{}(r)}
{\eval{t.\keyword{value}}{\srcset{}}{v}}
\and
\infer[ValueDerived]
{\eval{t}{\srcset{}}{r} \and t_d = \derset(r) \and \eval{t_d}{\srcset{}}{v}}
{\eval{t.\keyword{value}}{\srcset{}}{v}}
\end{mathpar}
\caption{Semantics for reactive term evaluation.}
\label{fig:term-semantics}
\end{figure}

\subsubsection{Semantics for interactions and device communication}

\begin{figure}
\centering
\begin{mathpar}
\infer*[left=Interact]
{
\intvar = \keyword{Interaction}((r_1, \ldots, r_n), l_{pre}, l_{post}, t_{exec})
\and
\intvar \in \intset
\\\\
\mathit{conflicts}(\intvar) \subseteq L
\and
\eval{(l_{pre}\ v_{arg})}{\srcset{}}{\keyword{true}} \and \eval{(t_{exec}\ v_{arg})}{\srcset{}}{(v_1, \dots, v_n)}
\\\\
\eval{(l_{post}\ v_{arg})}{\srcset{}'}{\keyword{true}} \and \srcset' = \mathit{update}(\srcset, (r_1, \ldots, r_n), (v_1, \dots, v_n))}
{(D_1 \mid \ldots \mid \device{}{} \mid \ldots \mid D_n) \xRightarrow[a]{v_{arg}} (D_1 \mid \ldots \mid \device{'}{} \mid \ldots \mid D_n)}
\and
\infer*[left=Sync]
{ \lockvar \subseteq L_s \and L_s' = L_s \setminus \lockvar \and L_r' = L_r \cup \lockvar \\\\
\srcset{}_r' = \textit{merge}(\srcset{}_r, \srcset{}_s)}
{(D_1 \mid \ldots \mid \device{_s}{_s} \mid \ldots \mid \device{_r}{_r} \mid \ldots \mid D_n) \xRightarrow[sync]{\lockvar} \\ (D_1 \mid \ldots \mid \device{_s}{_s'} \mid \ldots \mid \device{_r'}{_r'} \mid \ldots \mid D_n)}
\end{mathpar}
\begin{minipage}{.9\textwidth}
\begin{alignat*}{1}
\mathit{update}(\srcset, (r_1, \ldots, r_n), (v_1, \ldots, v_n))(r) &= \begin{cases}
    \mathit{merge}_r(\srcset(r), v_i) & \text{if } r = r_i\\
    \srcset(r) & \text{otherwise}
\end{cases}
\\
\mathit{merge}(\srcset_1, \srcset_2)(r) &= \mathit{merge}_r(\srcset_1(r), \srcset_2(r))
\\
\mathit{conflicts}&: \intset \rightarrow \mathbb{P}(\intset)
\end{alignat*}
\end{minipage}
\caption{Semantics for interactions and device communication.}
\label{fig:comm-semantics}
\end{figure}

Figure~\ref{fig:comm-semantics} presents the semantics for interactions
and device communication together with auxiliary functions.
We use three auxiliary functions
\emph{update}, \emph{merge} and \textit{conflicts}.
Function \emph{update} takes a set of source reactives $\srcset$, a tuple of reactive identifiers, a tuple of values and returns a new set with updated values.
The value of each source reactive is a CRDT -- we use a state-based model in our formalization:
When a source reactive is updated, the new value is computed by merging the update into the current value. The function $\mathit{merge}_r$ is defined by the CRDT stored in $r$.

Function \emph{merge} takes two stores $\srcset_1$ and $\srcset_2$ and merges them
by pair-wise merging the
values of each source reactive through $\mathit{merge}_r$.
The \emph{update} and \emph{merge} functions on stores are commutative, associative and idempotent,
because $\mathit{merge}_r$ also has these properties by the definition of CRDTs.
Notably, this implies a partial order on states where $\srcset \leq \mathit{update}(\srcset,(r_1,\ldots,r_n), (v_1,\ldots,v_n))$ and $\srcset_1 \le \srcset_2$ if $\mathit{merge}(\srcset_1, \srcset_2) = \srcset_2$. We will use this order to reason about cases where $\srcset_2$ includes all changes of $\srcset_1$.

For syntactic convenience, we lift the order to devices $D_{1/2} = \device{_{1/2}}{_{1/2}}$, where we write $D_1 \le D_2$ if $\srcset{}_1 \le \srcset{}_2$.

The \textit{conflicts} function takes an interaction and returns the set of interactions that conflict with it.
If an interaction $a_1$ has no conflicts, then $\mathit{conflicts}(a_1) = \emptyset$, if there is an interaction $a_2$ that conflicts with $a_1$,  then $\{\intvar_1, \intvar_2\} \subseteq \mathit{conflicts}(\intvar_1)$ and $\{\intvar_1, \intvar_2\} \subseteq \mathit{conflicts}(\intvar_2)$.
This ensures that whenever a device wants to execute an interaction with conflicts, it has to hold the locks for the interaction itself and for all the conflicting interactions.
The \textit{conflicts} function serves as an oracle to prevent concurrent execution of interactions.
Semantically, it defines the synchronization requirements of the program.
Specifying an empty conflict set for
every interaction, induces no synchronization, thus providing only causal consistency,
while specifying an interaction as conflicting with every other interaction yields
sequential consistency, disallowing any concurrent executions.
We show in Section~\ref{sec:chap:implementation} how to compute a conflict function
that prevents only those concurrent executions that would result in invariant violations on merged state.

We use a labelled transition system to model program execution,
where transitions are labelled by interactions or synchronizations, which both  occur  non-deterministically.
The \rulefont{Interact} rule defines the semantics for applying an interaction $\intvar$ with an argument $v_{arg}$.
The definition assumes that $l_{post}$, $l_{pre}$, and $t_{exec}$ are functions;
otherwise $\intvar$ is ill-defined and cannot be executed.
The execution of
$\intvar$ moves a device from state $\device{}{}$ to state $\device{'}{}$ if:
    i) the precondition $l_{pre}$ applied to
    $v_{arg}$  evaluates to $\keyword{true}$ before the execution,
    ii) the device executing $\intvar$ holds the locks of all $\mathit{conflicts(\intvar)}$,
    iii) the postcondition $l_{post}$ applied to $v_{arg}$ evaluates to $\keyword{true}$ after the execution,
    iv) evaluating $t_{exec}$ returns a tuple\footnote{Using a suitable encoding of tuples in the lambda calculus.} of new values $(v_1, \ldots, v_n)$ for each source reactive $(r_1, \ldots, r_n)$, which are used to update the store of source reactives $\srcset{}$.

The rule \rulefont{Sync} models communication between devices.
Synchronization happens between a sending device $D_s = \device{_s}{_s}$ and
a receiving device $D_r = \device{_r}{_r}$, where the former transfers state and locks to the latter.
A set of locks $\lockvar \subseteq L_s$ is removed from $D_s$ and added to the locks of $D_r$.
The state $\srcset{}_r$ is merged with the sent state $\srcset{}_s$.
By combining state updates and lock exchange in a single transition,
we ensure that for every interaction $\intvar$ that needs coordination,
a device $D_r$, which receives the locks for $\intvar$,
also receives the effects of the last application of $\intvar$.
Moreover, merging the state of all source reactives in a single step ensures
causal consistency for the entire state of the application,
not only for single reactives (i.e., single CRDTs).

An interaction only changes a single device, hence, we can abbreviate applications
of \rulefont{Interact} to $D \transition{a}{v_{arg}} D'$ to express that a device $D$ executed an interaction regardless of the state of the other devices.
Similarly, we write $D_s \transition{\mathit{sync}}{D_r} D_r'$ to express that $D_s$ was synchronized into $D_r$ producing $D_r'$
(and $D_s'$, which we do not need to state explicitly, as it is fully defined by the other 3 devices).

\subsection{Verifying Program Safety}
\langname{} guarantees that program execution is safe:
A program in a state where all safety invariants hold, transitions only into states where the invariants still hold.
A given program state that satisfies all safety invariants is called valid.
Formally:

\begin{definition}[Validity]
\label{def:validity}
Given
$P$ with invariants $\invset{}$ and devices $D_1, \dots, D_n$, we say that
$P$ is valid, written $\valid{P}$, if
$\valid{D_i}, i \in \{1, \dots, n\}$.
A device
$D = \device{}{}$ is \emph{valid} -- written $\valid{D}$ -- if for any
$\keyword{Invariant}(l) \in \invset{}$, $\eval{l}{\srcset{}}{\keyword{true}}$.
\end{definition}
\begin{definition}[Safety] \label{def:safety}
A program $P$ is \emph{safe}, if for any possible transition $P \Rightarrow P'$,
$\valid{P} \Rightarrow \valid{P'}$.
\end{definition}

To enable automatic and modular verification,
we break safety checking into two simpler properties on
invariants that can be automatically checked by Viper:
\emph{invariant preservation} and \emph{confluence}.
Invariant preservation ensures the safety of individual interactions,
whereas confluence (adapted from \authorcite[Bailis et al.]{bailis2014} and other works that build on the CALM theorem \cite{Alvaro2014, Hellerstein2019}) relates two interaction executions (of the same or different interactions).
We show how to mechanically prove these properties
using Viper in Section~\ref{sec:chap:implementation}.
The rest of this section formally introduces all mentioned properties and proves soundness of our approach
by showing that safety preservation follows from invariant preservation of all interactions,
and confluence of all interactions that are not marked as conflicting.

\begin{definition}[Invariant Preservation] \label{def:invariant-preservation}
An interaction $a$ is \emph{invariant preserving}, written $\preserving{a}$, if $(\valid{D} \wedge D \xRightarrow[a]{v_{arg}} D') \Rightarrow \valid{D'}$, i.e., given a valid device $D$ and some argument $v_{arg}$, the execution of $a$ in $D$ produces a valid device $D'$.
\end{definition}

\begin{definition}[Confluence] \label{def:invariant-confluence}
Interactions \(a_1\) and \(a_2\) are
\emph{confluent}, written $\confluent{a_1, a_2}$,
 if for any valid devices $D_i = \device{_i}{_i}$ and $D_j = \device{_j}{_j}$, and any argument values $v_1$ and $v_2$

\begin{alignat*}{1}
&D_i \xRightarrow[a_1]{v_1} \device{_i'}{_i} \wedge D_j \xRightarrow[a_2]{v_2} \device{_j'}{_j} \ \wedge \\
&D_{m_1} = \langle \mathit{merge}(\srcset_i, \srcset_j'), L_i \cup L_j \rangle \  \wedge \\
&D_{m_2} = \langle \mathit{merge}(\srcset_i', \srcset_j), L_i \cup L_j \rangle \  \wedge \\
&\mathit{merge}(\srcset_i', \srcset_j') = \srcset' \ \wedge \ L_i \cap L_j = \emptyset \\
\implies
&D_{m_1} \xRightarrow[a_1]{v_1} \langle \srcset{}', L_i \cup L_j \rangle \wedge\ \\
&D_{m_2} \xRightarrow[a_2]{v_2} \langle \srcset{}', L_i \cup L_j \rangle
\end{alignat*}

\end{definition}

Confluence states that applying interactions $a_1, a_2$ on two devices leads to the same results, no matter if a synchronization happens in-between or after the interactions.
Thus, concurrent execution of the two interactions leads to the same result as sequential execution.
Note that $a_1$ and $a_2$ are always trivially confluent if they affect disjoint subsets of source reactives and $a_1$ can never invalidate $a_2$'s precondition and vice versa.
In this case, sequential execution of the two interactions (in any order) always has the same result as concurrent execution.
We leverage this insight to reduce the amount of confluence proofs generated by our implementation in Section~\ref{step-2-graph-analysis}.

Given the definitions of invariant preservation and confluence, we now move on to prove soundness of \langname's execution model.

\begin{definition}[Initial Program]
The initial program $P^0$ has devices $(D^0_1 \mid \dots \mid D^0_n)$
that all have the same initial state of source reactives and all locks are with device $D^0_1$, i.e.,
$D^0_1 = \langle \srcset{}, \intset{} \rangle \wedge D^0_i = \langle \srcset{}, \emptyset{} \rangle, \forall i \in \{2, \ldots, n\}$.
\end{definition}

\begin{lemma}[Correct locking] \label{lemma-correct-locks}
The locking mechanism ensures for any program execution $P^0 \Longrightarrow \dots \Longrightarrow P^m$ that conflicting interactions are sequentially ordered.
Specifically, for any two conflicting interactions $a_1$, $a_2$ with transitions $\device{_1}{_1} \transition{a_1}{v_1} \device{'_1}{_1}$ and $\device{_2}{_2} \transition{a_2}{v_2} \device{'_2}{_2}$, and $\mathit{conflicts}(a_1) \cap \mathit{conflicts}(a_2) \ne \emptyset$,
the starting state of either $a_1$ or $a_2$ must include all changes produced by the other:  $\srcset{}'_2 \le \srcset{}_1$ or $\srcset{}'_1 \le \srcset{}_2$.
\end{lemma}

\begin{proof}
We first show by induction, that any program state that is reachable from the initial program $P_0$ assigns each lock to exactly one device.
This is true for the initial program, interact transitions do not modify the lock assignment, and the sync rule removes the same set of locks from the sending devices that are added to the receiving device.

Then, we prove that non-overlapping locks ensure Lemma~\ref{lemma-correct-locks}.
Conflicting interactions require the same lock $l \in \mathit{conflicts}(a_1) \cap \mathit{conflicts}(a_2)$.
This lock must be present on both devices executing the conflicting interactions, i.e., $l \in L_1$ and $l \in L_2$.
Thus, the lock must be transferred from $D_1' = \device{'_1}{_1}$ to $D_2 = \device{_2}{_2}$ (or symmetrically from $D_2'$ to $D_1$), using any number of transitions.
By commutativity, associativity and idempotence of the $\mathit{merge}$ function, each \rulefont{sync} transition ensures that the receiving device $D_r$
has a state $\srcset{}_r'$ with $\srcset{}_s \le \srcset{}_r'$ where $\srcset{}_s$ denotes the state of the sending device.
Similarly, every \rulefont{interact} transition $\device{}{} \transition{a}{v} \device{'}{}$ ensures that $\srcset{} \leq \srcset{'}$.
This implies that either $\srcset{}_2' \le \srcset{}_1$ or $\srcset{}_1' \le \srcset{}_2$.
\end{proof}

\begin{theorem}[Soundness]
Given a
program $P$ with interactions $\intset{}$ and invariants $\invset{}$, if the verifier has shown that
(i) $\forall  a \in \intset{}$. $\preserving{a}$,
and (ii) $\forall  \intvar_1, \intvar_2 \in \intset$,
$\confluent{\intvar_1, \intvar_2}$ $\lor$ $\intvar_2 \in \mathit{conflicts}(\intvar_1)$,
then
$P$ is safe.
\end{theorem}

\begin{proof}
To prove soundness, we show that for any sequence $\mathcal{C}$
of transitions
$P^0 \Longrightarrow \dots \Longrightarrow P^m$,
$\valid{P^0}$ implies $\valid{P^m}$.
To show $\valid{P^m}$, we show that $\forall D^m_i \in P.\ \valid{D^m_i}$.
We show that an arbitrary $D^m \in P^m$
is valid,
by showing that
we can construct a serialization order $\mathcal{S}_i$
of interaction transitions $D^0_1 \xRightarrow[\intvar_1]{v_1} \ldots \xRightarrow[\intvar_l]{v_l} D'$
that starts from the initial device state with all locks, consists exclusively of
applications of the \rulefont{Interact} rule, and
yields a device $D' = \device{^m_i}{}$ with the same state of
source reactives $\srcset^m_i$ as
$D^m_i$.\footnote{Note that the length of $\mathcal{C}$ and $\mathcal{S}_i$ may differ, and $\mathcal{S}_i$ does not necessarily include all interactions of $\mathcal{C}$.}
Given such $\mathcal{S}_i$ and $\forall a \in \mathcal{S}_i.\ \preserving{a}$
(which follows from premise (i) of the Theorem),
we can conclude that $\valid{D'}$. In turn, $\valid{D'}$ implies $\valid{D^m_i}$
because validity only depends on $\srcset^m_i$.
In other words, we show that every possible device state that
is a result of interactions and synchronizations, could
also be constructed through a sequence of only local
interactions.
Since every interaction on itself is \emph{invariant preserving} (premise (i)), this implies that the program is safe.

\emph{Serial order construction.}
Let $D^m \in P^m$ be an arbitrary device
that we choose to serialize.
As a convention, we use $D$ (without the superscript) to refer to that device,
in any program state belonging to $\mathcal{C}$ or $\mathcal{S}$.
e.g., an interaction on $D$ would be a step of the form $D^k \xRightarrow[\intvar_1]{v} D^{k+1}$.
We construct the serial order $\mathcal{S}$ for $D$ stepwise from the concurrent order $\mathcal{C}$,
by picking transitions at the end of $\mathcal{C}$ and prepend them to $\mathcal{S}$.
We say that the last (rightmost) transition in $\mathcal{C}$ is the focused transition
$T$.
Below we consider all possible rule applications that trigger $T$ in a case-by-case way.

\textbf{Case 1:} $T$ is an interaction
$D \transition{a}{v} D'$.
We prepend $T$ to $\mathcal{S}$ and discard it from $\mathcal{C}$.

\textbf{Case 2:} $T$ is an interaction $D_i \transition{a}{v} D_i'$, $D_i \ne D$.
We discard $T$ from $\mathcal{C}$, because by being the last transition in $\mathcal{C}$,
it does not affect $D$.

\textbf{Case 3:} $T$ is a synchronization $D_s \transition{\mathit{sync}}{D_r} D_r' $ from any $D_s$ (possibly $D$) to $D_r \ne D$.
We discard $T$ from $\mathcal{C}$ because
we know that $D_r'$ will not synchronize with $D$ after $T$.

\textbf{Case 4:} $T$ is a synchronization $D_i \transition{\mathit{sync}}{D} D'$ from any $D_i \ne D$ to $D$.
To handle this case, we consider possible states of $D_i$ and $D$ before $T$ occurred,
especially whether devices had concurrent changes since their last synchronization:

\textbf{Case 4.1:} $D_i \le D$, i.e., there are no changes to $D_i$ compared to $D$.
We discard $T$, because it only transfers some lock, and locks are irrelevant for the final serialization order.

\textbf{Case 4.2:} $D \le D_i$, i.e., there are no changes to $D$ compared to $D_i$.
Similar to case 4.1, but now all relevant changes are on $D_i$.
Because $T$ is a sync, we know that the state of $D$ after $T$ is equal to $D_i$ (except locks).
We discard $T$ and continue the construction with $D_i$ as the chosen device.

\textbf{Case 4.3:}
Both $D$ and $D_i$ have concurrent changes
and the transition that produced the state of $D$ preceding $T$ is the application of an interaction $a$.
We know that any interaction $a$ is either confluent to any other interaction (including itself) or the other interaction is included in $\mathit{conflicts}(a)$ (premise (ii) of the Theorem).
Given Lemma~\ref{lemma-correct-locks}, this implies
that all interactions on $D_i$ that are concurrent to $a$ are confluent with $a$.
Thus, we can use the confluence definition to reorder $\mathcal{C}$ to put the application of $a$
at the rightmost end of $\mathcal{C}$, directly after the synchronization $T$. After the reordering,
we handle $a$ according to case 1.

\textbf{Case 4.4:}
Both $D$ and $D_i$ have concurrent changes and the state of $D$
was produced by a synchronization $T_k$, $D_k \transition{\mathit{sync}}{D'} D$, that synchronizes a third device $D_k \ne D_i$ into $D$.
In other words, the effect of the last two transitions is a merge between $D_k$, $D_i$, and $D$ and we (potentially) have concurrent changes from all three for which we must find a single serialization order.
On a high-level, we (arbitrarily) choose to order concurrent interactions on $D$ first, $D_i$ second, and $D_k$ third by changing the transitions to first synchronize $D_k$ into $D_i$ and then $D_i$ into $D$.
In other words, we change $T_k$ (the sync from $D_k$ to $D$) in $\mathcal{C}$ to $T'_k = D_k \transition{\mathit{sync}}{D'} D_i$.%
\footnote{This disregards changes in transferred locks, but at this point those are irrelevant for the serialization order.}
This does not change the result state of $D$ -- it is still a result of merging the states of $D$, $D_i$, and $D_k$, where the change in merge order does not matter because merging is associative and commutative.

\emph{Closing remarks.}
Our construction terminates, because (a) there is a finite amount of transitions in $C$
and (b) each case above, except 4.4, reduces the size of $C$. Case 4.4 does not reduce the size
of $\mathcal{C}$, but it is impossible to indefinitely repeat it, as this would entail that there are
indefinitely many synchronizations from $D_k$ into $D$, which is impossible as $\mathcal{C}$ is assumed to be finite.

Also, the construction
handles all situations, because cases 1 and 2 cover all possible applications of interactions,
while cases 3 and 4 cover all applications of synchronizations -- these are the only two rules that produce transitions.
We know that the distinction in case 4 is complete, because cases 4.1 and 4.2 handle all situations where there are no concurrent changes\footnote{Technically, both devices could apply confluent interactions that lead to the same result state -- these are irrelevant due to idempotence of merges and thus can be discarded from the serialization.} on one of the devices.
If there are concurrent changes, then cases 4.3 and 4.4 again exhaustively cover that situation.
In particular, the one exclusion from case 4.4, specifically the situation that $D_k = D_i$ is covered by case 4.2, because then $D_i$ is synchronized into $D$ twice, without any changes to $D$ in between.

In summary, by successively replacing and discarding transitions according to the cases above, we can generate a sequential order for any device $D_i^m \in P^m$, from which follows (using invariant preservation) that $\valid{D_i^m}$.
As we can construct such a serialization (which must not be the same) for all devices, we know that $\valid{P^m}$.

\end{proof}

\section{Implementation}
\label{sec:chap:implementation}

\begin{figure*}
\centering
\includegraphics[width=0.81\textwidth]{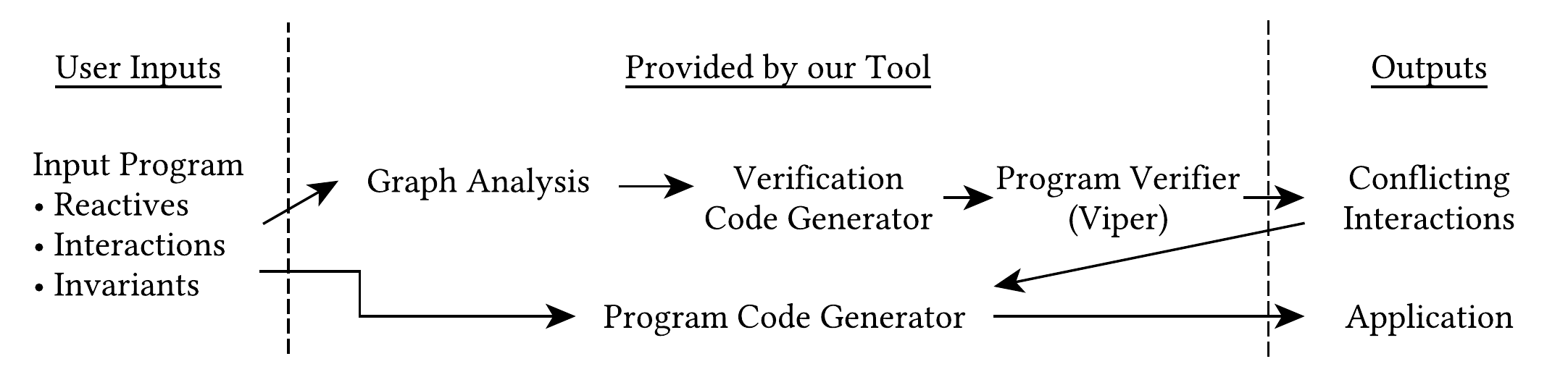}
\caption{Overview of \langname's automated compilation and verification procedure.}
\label{fig:procedure}
\end{figure*}

Figure~\ref{fig:procedure} depicts the architecture of \langname{}'s verifying compiler.
The input to the compiler is a program with its specifications expressed by the invariants,
e.g., the program in Listing~\ref{lst:calendar}.
The output consists of the conflicting interactions and a safe executable program.
To verify safety, we
prove invariant preservation (Definition~\ref{def:invariant-preservation}) for all interactions (a failed proof results in a compilation error) and try to prove invariant confluence (Definition~\ref{def:invariant-confluence}) for all pairs of interactions. We employ an analysis of the data-flow graph to minimize confluence
proof obligations  to those invariant pairs that may actually conflict.
Non-confluent interaction pairs are included in the conflict output.
We use Viper for the verification step, but any other verification could be used
and would still benefit from our minimization of the proof obligations.
The implementation uses REScala, CRDTs, and a token-based locking protocol for generating a
safe executable program. But they could be replaced by other implementations
of data-flow programming, eventual consistency,
and consensus with the same guarantees.
The rest of this section describes the pipeline from Figure~\ref{fig:procedure}
in detail -- from left to right, top to bottom.

\subsection{Graph Analysis}
\label{step-2-graph-analysis}

Checking all pairs of interactions for confluence would result in an exponential amount of proof obligations.
To avoid this, we employ a graph analysis to quickly detect pairs of interactions that cannot conflict, because they change completely separate parts of the data-flow graph.
The graph analysis algorithm (Figure~\ref{alg:dataflow}) checks every interaction to determine its reachable reactives:
The source reactives that are affected by the interaction together with all transitively derived reactives that depend on them (Line~\ref{line:reaches}).
Next, the algorithm determines the \emph{overlaps} between interactions and invariants (Line~\ref{line:overlaps}).
An interaction and an invariant \emph{overlap} if any reactive occurring in the invariant is part of the interaction’s reachable reactives.
Two interactions \emph{overlap} if there is at least one invariant they both overlap with.
Only overlapping interactions produce proof obligations.

\begin{algorithm}[t]
\small
\SetKwFunction{RelevantSubgraph}{relevantSubgraph}
\Proc{overlapAnalysis}{
    \For{\Each interaction}{
        determine \textit{overlappingInvariants(interaction)}\;
    }
}
\BlankLine
\Fn{overlappingInvariants(interaction)}{\label{line:overlaps}
    subgraph $\leftarrow$ \textit{reaches(interaction)}\;
    $\mathcal{I}$ $\leftarrow$ $\emptyset$\;
    \For{\Each invariant}{
        $\mathcal{R}$ $\leftarrow$ all reactives that this invariant mentions\;
        \If{$\mathcal{R} \cap \mathit{subgraph} \neq \emptyset$ }{
            $\mathcal{I}$ $\leftarrow$ $\mathcal{I}$ $\cup$ $\{$ \textit{invariant} $\}$\;
        }
    }
    \Return $\mathcal{I}$
}
\BlankLine
\Fn{reaches(interaction)}{\label{line:reaches}
    $\mathcal{R}$ $\leftarrow$ all reactives that this interaction modifies\;
    subgraph $\leftarrow$ $\mathcal{R}$\;
    \For{each $r \in \mathcal{R}$}{
        children $\leftarrow$ all reactives that transitively depend on $r$\;
        subgraph $\leftarrow$ subgraph $\cup$ children\;
    }
    \Return subgraph
}
\caption{Pseudocode of the algorithm used to determine which invariants overlap with a transaction.}
\label{alg:dataflow}
\end{algorithm}

For illustration, consider
the \lstinline{add_work} interaction. It modifies the \lstinline{work} reactive, and -- transitively -- \lstinline{all_appointments}.
Hence, the reachable reactives are \(\{\mathit{work}, \mathit{all\_appointments}\}\) and only the first but not the second invariant in Listing~\ref{lst:calendar} overlaps. Thus, neither the \lstinline{remaining_vacation} reactive, nor the invariant on this reactive will be part of the proof obligation for the \lstinline{add_work} interaction.

\subsection{Automated Verification}
\label{step-3-verification}

\subsubsection{Translation to Viper's Intermediate Language.}
\begin{lstlisting}[float=*, language=Sil, escapechar={§}, caption={Viper representation of the graph and invariants of the calendar example.}, label=lst:viper-header]
// GRAPH
// sources
field work: AWSet[Appointment] §\label{line:graph-start}§
field vacation: AWSet[Appointment]
// derived
define all_appointments(work, vacation) toSet(work) union toSet(vacation) §\label{line:graph-derived}§
define remaining_vacation(vacation) 30 - countDays(toSet(vacation)) §\label{line:graph-end}§

// INVARIANTS
define inv_1(all_appointments) forall a: Appointment :: a in all_appointments ==> get_start(a) < get_end(a) §\label{line:viper-invariants}§
define inv_2(remaining_vacation) remaining_vacation >= 0
\end{lstlisting}

Listing~\ref{lst:viper-header} illustrates how we represent the data-flow graph and the safety invariants of our calendar example (Listing~\ref{lst:calendar}) in Viper's intermediate verification language~\cite{MuellerSchwerhoffSummers16}.
We represent the data-flow graph as a mutable object with one field per source reactive (Line~\ref{line:graph-start}).
In Viper, objects are implicitly defined by which fields the program accesses.
Derived reactives (Line~\ref{line:graph-derived}) and invariants (Line~\ref{line:viper-invariants}) are expressed as Viper macros -- pure functions that describe the invariant or body of the reactive, and receive the reactives they depend on as function inputs.

Given these definitions, we synthesize one Viper \emph{method} per interaction, based on the formal evaluation semantics in Section~\ref{sec:programming-model}.
Viper verifies programs on a per-method basis where methods represent a sequential computation annotated with pre- and postconditions.
Listing~\ref{lst:viper-preservation} shows the Viper method for the \lstinline{add_work} interaction.
Pre- and postconditions of interactions are simply translated to pre- and postconditions of the Viper methods (Line~\ref{line:pres-preconditions} and \ref{line:pres-postconditions}) while the \lstinline{executes} part of each interaction is represented by the method body (Line~\ref{line:pres-body}).
Reading derived reactives is modeled by inlining the respective Viper macro, which results in our big-step evaluation semantics.
Additionally, we include every overlapping invariant as pre- and postconditions (Lines~\ref{line:pres-invariants1} and \ref{line:pres-invariants2}) so that the verifier can prove invariant preservation.
Evaluating invariants is expressed by nested application of the previously defined macros (Listing~\ref{lst:viper-header}) to the source reactives, corresponding to the propagation of values though the data-flow graph.
Viper uses explicit permissions for shared state, thus we explicitly pass a reference to the data-flow graph (Line~\ref{line:graph-ref}); it declares write permissions for all reactives that are modified by the interaction and read permissions for the reactives that are only accessed as part of the invariants (Line~\ref{line:pres-permissions})\footnote{
Viper uses \emph{fractional permissions} where an \lstinline{acc(n)} statement with any $n < 1$ corresponds to a read-permission and a statement with $n = 1$ corresponds to a write-permission.}.

\begin{lstlisting}[float=t, language=Sil, escapechar={§}, caption={Viper representation of the \lstinline{add_work} interaction.}, label=lst:viper-preservation]
method add_work( §\label{line:pres-method-start}§
  graph: Ref, §\label{line:graph-ref}§
  appointment: Appointment)
// permissions
requires acc(graph.work, 1) && requires acc(graph.vacation, 1/2) §\label{line:pres-permissions}§
// preconditions
requires !(appointment in toSet(graph.work)) §\label{line:pres-preconditions}§
requires get_start(appointment) < get_end(appointment)
// overlapping invariants
requires inv_1(all_appointments(graph.vacation, graph.work)) §\label{line:pres-invariants1}§
// postconditions
ensures acc(graph.work, 1) && acc(graph.vacation, 1/2) §\label{line:pres-postconditions}§
ensures appointment in toSet(graph.work)
ensures size(graph.work) == old(size(graph.work)) + 1
ensures toSet(graph.work) == old(toSet(graph.work)) union Set(appointment)
// overlapping invariants
ensures inv_1(all_appointments(graph.vacation, graph.work)) §\label{line:pres-invariants2}§
{
  graph.work := add(graph.work, appointment) §\label{line:pres-body}§
}
\end{lstlisting}

\subsubsection{Proving Invariant Preservation and Confluence.}
\label{proving-preservation}
Given the Viper encoding of each interaction -- which includes overlapping invariants -- Viper directly outputs verification results that correspond to invariant preservation (Definition~\ref{def:invariant-preservation}).
Any verification errors at this stage are errors in the supplied specification or programming bugs and must be addressed by the programmer.

To check
for confluence the compiler creates a Viper method for each overlapping pair of interactions.
Such a method models the specification for invariant confluence (Definition~\ref{def:invariant-confluence}).
Any verification errors here indicate that the two interactions are non-confluent.
As our soundness proof requires either confluence of interactions, or their inclusion in the set of conflicting interactions, safety is ensured by marking all non-confluent pairs of interactions as conflicting.
Developers should still check the conflicts to ensure that they are
not due to a bug or specification error.

\subsection{Synchronization at Runtime}
\label{sec:synchronization}

Our compiler generates an executable application by converting
the data-flow graph to a distributed REScala program~\cite{Mogk2018, Mogk2019}.
REScala supports all reactive features we require and integrates well with our CRDT-based replication, but has no mechanism for synchronization.
\langname's formal synchronization semantics (cf. Section~\ref{sec:programming-model}) could be implemented using any existing form of coordination, such as a central server, a distributed ledger, a consensus algorithm, or a distributed locking protocol.
Which choice is suitable, depends on the target application, network size, and the reliability of the chosen transport layer.
We use a simple distributed locking protocol for our prototype implementation:
Each interaction has an associated lock (represented as a simple token).
Whenever a device wants to execute an interaction,
it acquires the tokens of all conflicting interactions.
If multiple devices request the same token concurrently, the token is given to the device with the lowest ID that requested it.
This ensures deadlock freedom; fairness is left for future work.
After performing the interaction, the resulting state changes are synchronized with the other devices and the tokens are made available again.
Timeouts ensure that whenever a device crashes or becomes unavailable for a longer period of time, its currently owned tokens are released and any unfinished interactions by the device are aborted.

\section{Evaluation}
\label{sec:evaluation}

Our evaluation aims to validate two claims about \langname{}'s programming model:
\begin{description}
\item[C1:]
It facilitates the development of safe
local-first software.
\item[C2:] It enables an efficient and modular verification of safety properties.
\end{description}

We base our validation on two case studies.
First, we implemented the standard TPC-C benchmark~\cite{TPCCSpecification}
as a local-first application in \langname.
This case study enables comparing \langname{}'s model with traditional database-centered development of
data processing software and showcasing the benefits of \langname's
verifiable safety guarantees on standard
\emph{consistency conditions}.
Second, we implemented the running calendar example (Section~\ref{sec:goals})
using Yjs \cite{DBLP:conf/group/NicolaescuJDK16}. This case study allows comparing \langname{}
with an existing framework for local-first applications that we consider a representative of the state-of-the-art.

\subsection{Does \langname{}
facilitate the development of safe
local-first software?}

\subsubsection{Local-first TPC-C}
\label{sec:tpc-c}

TPC-C models an order fulfillment system with multiple warehouses in different districts,
consisting of five \emph{database transactions}
alongside twelve \emph{consistency conditions}.
We implemented TPC-C in \langname\ by mapping database tables to source reactives and derived database values to derived reactives.
Each database transaction was modelled as a \langname\ interaction.

In the \langname{} implementation of TPC-C,
each warehouse holds a local copy of the data on which transactions -- modeled as interactions -- are executed  before being synchronized with other warehouses.

\emph{Reactives.}
Figure~\ref{fig:tpc-c-graph} shows the reactive graph of the application.
The structure roughly follows the database structure described in TPC-C.
We represent each database table by a source reactive and model derived database values as derived reactives.
For example, the \lstinline{DistrictYTD} reactive in Listing~\ref{lst:districtYTD} represents the year-to-date (YTD) balance of districts.
After reading the \lstinline{districts} reactive (Line~\ref{line:read-dists}), we perform
the following steps for each district:
We read the relevant entries in the payment history (Line~\ref{line:payment-entries}),
calculate the sum of the YTD values (Line~\ref{line:ytd-sum}),
and return the result as a single entry mapping from district to YTD (Line~\ref{line:dist-tuple}).

\begin{lstlisting}[caption={The \langname{} representation of the derived \lstinline{DistrictYTD}
 reactive (simplified).}, float=t, language=fr, escapechar={§}, label=lst:districtYTD]
val districtYTD: Derived[Map[District, YTD]] = Derived {
  districts.map{d => §\label{line:read-dists}§
    val payments = paymentHistory.filter(ph =>§\label{line:payment-entries}§ sameDistrict(ph, d))
    val ytd = payments.sum() §\label{line:ytd-sum}§
    (d, ytd) §\label{line:dist-tuple}§
  }.toMap() }
\end{lstlisting}

\emph{Interactions.}
We implement TPC-C transactions as interactions.
For illustration, consider the \lstinline{payment} interaction
in Listing~\ref{lst:payment-trans}, which is applied whenever a customer pays a certain amount to the system.
Payments are not associated to a specific order and are simply stored in the payment history.
Table modifications in TPC-C have multiple arguments, which we encapsulate into an argument of type \lstinline{PaymentArgs} (Line~\ref{line:payment-trans-def}).
Applying the interaction (Line~\ref{line:cust-retrA}) retrieves the customer object matching the payment arguments
by executing a function that accesses the \lstinline{customers} reactive (Line~\ref{line:cust-retrB}).
Line~\ref{line:history-transform} updates and returns the new history.
The preconditions (Line~\ref{line:payment-pre}) encodes assumptions about the arguments,
notably that a customer for whom we add the payment actually exists,
and the postcondition (Line~\ref{line:payment-post}) describes the effect of adding a new payment.

\begin{lstlisting}[caption={The \langname{} representation of TPC-C's payment interaction (simplified).}, float=t, language=fr, escapechar={§}, label=lst:payment-trans]
val payment: Unit = Interaction[PaymentHistory][PaymentArgs] §\label{line:payment-trans-def}§
.modifies(paymentHistory)
.executes{ph => pa => {§\label{line:cust-retrA}§
   val c: Customer = findCustomer(customers, pa.c_id, pa.c_last) §\label{line:cust-retrB}§
   val id: Id = getId(c)
   ph.add(new_payment(id, pa))}§\label{line:history-transform}§
}
.requires{ph => pa => pa.h_amount > 0 && exists c: Customer ::§\label{line:payment-pre}§
   (c in customers.toSet && (c.id == pa.c_id || c.last == pa.c_last))}
.ensures{ph => pa => exists p: Payment ::
    toSet(ph) == old(toSet(ph)) union Set(p)}
.ensures{ph => pa => old(ph.toSet).subset(ph.toSet)}§\label{line:payment-post}§
\end{lstlisting}

\begin{lstlisting}[caption={TPC-C's consistency condition 5 expressed as a \langname{} invariant.}, float=t, language=fr, escapechar={|}, label=lst:cc5]
invariant forall o: Order :: o in Orders ==> (
    o.carrier_id == 0 <==>
    (exist no: NewOrder :: no in NewOrders && no.id = o.id) )
\end{lstlisting}

\emph{Invariants.}
TPC-C defines 12 consistency conditions, of which
9 are consistency constraints between tables and derived tables.
Their correctness is automatically ensured by \langname{} without further specification.
For example, the consistency condition 9~\cite{TPCCSpecification}:
\begin{quote}
Entries in the DISTRICT and HISTORY tables must satisfy the relationship:\\
D\_YTD = sum (H\_AMOUNT) for each district defined by \\ (D\_W\_ID, D\_ID) = (H\_W\_ID, H\_D\_ID).
\end{quote}
directly corresponds to the definition of the \texttt{DistrictYTD} reactive in Listing~\ref{lst:districtYTD} and is thus always true by design.
Only the remaining 3 conditions require translating the natural language specification into first-order logic formulae to be used as invariants.
As an example, consider consistency condition 5 from the TPC-C specification~\cite{TPCCSpecification}:
\begin{quote}
For any row in the ORDER table, O\_CARRIER\_ID is set to a null value if and only if there is a corresponding row in the NEW-ORDER table [\ldots].
\end{quote}
This condition cannot be represented as a derived reactive because the carrier ID is not derived from other values but set explicitly whenever an order is shipped.
Listing~\ref{lst:cc5} displays the encoding
in \langname{} -- an (almost) literal translation of the specification.

\emph{Comparison.}
The implementation effort for the core functionality of TPC-C in \langname\ is comparable to a traditional design relying on a relational database model.
While modelling the application using reactives might require some adaption from developers not familiar with data-flow programming, we found that using derived reactives led to a more concise and less error-prone design when compared to storing derived values in separate tables.
This observation is supported by the fact that we only need to explicitly address 3 out of 12 consistency conditions of
TPC-C.
We were able to phrase the remaining 3 conditions
as invariants by directly translating the natural language formulations into logical
specifications.
To prove them, we additionally needed to specify pre- and postconditions of interactions corresponding to transactions (see Figure~\ref{lst:payment-trans}).
Other than that, \langname{} relieves the TPC-C developer
from any considerations of
transaction interleavings
that could potentially violate the conditions as well as
from implementing the synchronization logic, both tedious and
error-prone processes.
Moreover, unlike TPC-C, which cares only about the consistency of
the database,
\langname{} treats consistency constraints uniformly from (shared) state to UI.
This
is enabled by the reactive programming paradigm, which also
guarantees 9 out of 12 consistency conditions by design.

\subsubsection{Yjs-based Calendar}
\label{sec:yjs}

\begin{figure}
\centering
\def\svgwidth{1.5\columnwidth}
\scalebox{.66}{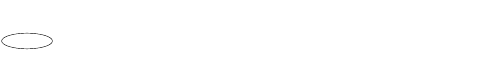}
\caption{The data-flow graph of the TPC-C benchmark application.}
\label{fig:tpc-c-graph}
\end{figure}

\begin{figure}[t]
\centering
\begin{minipage}[t]{0.48\linewidth}
\begin{lstlisting}[caption={Defining source and derived\\variables in Yjs.}, language=js, escapechar={§}, label=lst:calendar-yjs, xleftmargin=0pt, xrightmargin=10pt]
const ydoc = new Y.Doc()
let work = ydoc.getMap('work');§\label{line:yjs-sources1}§
let vacation = ydoc.getMap('vacation');§\label{line:yjs-sources2}§
let all_appointments;
let remaining_vacation = 30;

work.observe(ymapEvent => {§\label{line:yjs-callback1}§
 all_appointments = getMap(work, vacation);
})

vacation.observe(ymapEvent => {
 let days_total = getTotalVacDays(vacation);
 remaining_vacation = 30 - daysTotal;
 all_appointments = getMap(work, vacation);
})§\label{line:yjs-callback2}§
\end{lstlisting}
\end{minipage}
\quad
\begin{minipage}[t]{0.48\linewidth}
\begin{lstlisting}[xleftmargin=10pt, xrightmargin=0pt, caption={Adding appointments in Yjs.}, language=js, escapechar={§}, label=lst:yjs-interactions]
function addAppointment(calendar, appointment) {
 if(appointment.start < appointment.end){§\label{line:yjs-precond1}§
   calendar.set(appointment.id, appointment);
 }
}

function addVacation(appointment)
{
  let days =
    appointment.getDays();
  if(remainingVacation < days){§\label{line:yjs-precond2}§
    console.log("Sorry, no vacation left!");
  }
  else{
    addAppointment(vacation, appointment)
  }
}
\end{lstlisting}
\end{minipage}
\end{figure}

We now compare the \langname\ implementation of the distributed calendar to an implementation using the state of the art local-first framework Yjs~\cite{DBLP:conf/group/NicolaescuJDK16}. Like other solutions for local-first software,
Yjs uses a library of CRDTs (usually maps, sets, sequences / arrays, and counters) composed into nested trees -- called a \emph{document} -- used to model domain objects.

\emph{Source and Derived Variables.}
For illustration, consider Listing~\ref{lst:calendar-yjs}, showing how one could implement the domain model of the calendar application.
Lines~\ref{line:yjs-sources1} and \ref{line:yjs-sources2} initialize two CRDTs for the work and vacation calendar.
Yjs has no abstraction for derived values and only provides callbacks for reacting to value changes,
e.g., Lines~\ref{line:yjs-callback1}-\ref{line:yjs-callback2} declare callback methods that update the derived variables in case the Yjs document changes.

\emph{Safety Guarantees.}
Using callbacks to model and manage complex state that changes  both in time and in space
has issues. It requires that developers programmatically update the derived values
once the sources get updated, via local interactions or on receiving updates from other
devices, with no guarantees that they do so consistently.
It yields a complex control-flow and requires intricate knowledge of the execution
semantics to ensure atomicity of updates, let alone to enforce application-level safety properties.
Frameworks like Yjs do not offer support for application invariants and thus force developers to integrate custom safety measures at each possible source of safety violations.
As an example, consider Listing~\ref{lst:yjs-interactions}, showing how the \lstinline{addVacation} interaction could be implemented in Yjs.
Lines~\ref{line:yjs-precond1} and \ref{line:yjs-precond2} check
the preconditions of the interaction,
but as discussed in Section~\ref{sec:synchronization-points}, such local checks are insufficient to maintain safety. %
In general, global reasoning without tool support quickly becomes infeasible,  and even more
so for designs based on callbacks.
Once updates inducing invariant violations are propagated, it is difficult to ``undo'' them.
Thus, programmers are required to either provide and coordinate compensation
actions or to integrate mechanisms for strong synchronization on top of
CRDT-based eventual consistency.
Both approaches are difficult to get right, putting safety and/or availability at risk.

In summary, while the replication capabilities of systems like Yjs are valuable for local-first applications,
these systems still require the developer to do
state management manually.
The prevailing use of callbacks and implicit dependencies
 makes reasoning about the code challenging for both developers and automatic analyses.
In contrast, \langname\ allows declarative definitions of derived values,
with positive effects on reasoning~\cite{Salvaneschi2017empirical, Dinser2021Thesis}.
Moreover, \langname\ integrates application invariants as explicit language constructs, which allows for a modular specification and verification and relieves developers from having to consider every involved interaction whenever the specification changes.

\subsubsection{On the Ergonomics of \langname's Programming Model}

While a detailed evaluation of \langname's ergonomics is out of the scope of this paper, in this subsection, we briefly outline why we expect  \langname's programming model to be easier to use than current models for local-first software. \langname\ coherently combines the following main abstractions: \emph{reactives}, \emph{interactions} and \emph{invariants}, which provide programmers with a principled way to think about the data flow, (user) interactions, and safety guarantees of their program respectively.
None of these abstractions is in itself unfamiliar:
Firstly, CRDTs -- the basis for \langname's reactive data model -- are popular data structures for modelling decentralized local-first software, available in libraries like Automerge~\cite{automerge}, Collabs~\cite{weidner2022} or Yjs~\cite{DBLP:conf/group/NicolaescuJDK16}.
Secondly, reactive programming is a popular paradigm especially for interactive web applications and is present in widely used frameworks like React~\cite{reactjs}.
Past research has shown in user studies that functional reactive programming (FRP), the variant of reactive programming used in \langname\,  can improve program comprehension when compared to a classic callback-driven programming style as it is present in vanilla JavaScript~\cite{Salvaneschi2017empirical, Dinser2021Thesis}.
\langname\ combines CRDTs and FRP in one programming model. 
By doing so, \langname\ makes it easier to reason about properties of 
compositions of CRDTs involved in an interaction. Finally, \langname\ 
complements the CRDT-FRP combination with a \emph{correctness-by-construction} 
approach~\cite{kourie2012} through its \emph{invariants} and its verifying compiler. 

By making invariants a first-class programming construct, \langname\ 
encourages programmers to write programs alongside their specification.
Given the specification, the verifying compiler helps programmers to refine their program until it meets their expectations.
Depending on the complexity of the application, this verification step introduces some additional specification burden when compared to approaches that do not rely on verification and thus offer weaker guarantees.
In our experience, the specification burden of \langname's model is modest and comparable to other deductive verification approaches~\cite{Barnett2006,pearceWhileyPlatformResearch2013,leino2010} (see Section~\ref{sec:tpc-c} for examples of the required specifications). It is also worth noting that in \langname\, 
specifications and invariants are optional; they can thus be used depending on the desired strength of the correctness guarantees. 
This allows for an incremental workflow where specifications are added step-by-step to an existing codebase.

Summarizing our observations from Sections \ref{sec:tpc-c} and \ref{sec:yjs}, 
\langname's programming model proved very effective for our case studies 
and facilitated the development when compared to a classic relational data model (TPC-C) and existing local-first frameworks (Yjs).
In our experience, the FRP model simplifies state management and composition of CRDTs, 
while the verification process helps port safety-critical applications like TPC-C to a decentralized setting.

\subsection{Does \langname{} enable efficient and modular verification of safety properties?}

Safety invariants in \langname{} are \emph{global} in the sense that they must not be violated by \emph{any} part of the program.
But their enforcement is based on verifying individual local properties.
We limit the need for verification to potential conflicts that we derive from the reactive data-flow graph.
This optimization requires no further reasoning from the programmer and relies solely on the properties of the programming model.
Programmers can add new functionality to the application (i.e., specify interactions) and only have to reason about the properties of that new functionality (i.e., specify its invariants) and the system ensures global safety -- at only the cost of the amount of overlap with existing functionality. This allows a modular programming style where invariants and interactions are written by different developers and changes to the program are made incrementally.

To empirically
evaluate the performance of \langname{}'s verifier,
we quantify how long it takes to verify different
combinations of interactions and invariants of our two case studies.
The results are shown in Table~\ref{tbl:verification-time}.
The calendar example has two additional types of interactions,
 which we have not shown in Section~\ref{sec:goals}: removing and changing calendar entries.
This leads to a total of 6 interactions (3 per calendar reactive).
As explained in the previous section, for TPC-C we only had to verify consistency
conditions 3, 5, and 7.
The benchmarks were performed on a desktop PC with an AMD Ryzen 7 5700G CPU and 32 GB RAM using Viper's \emph{silicon} verification backend (release v.23.01)~\cite{viperprojectSilicon}.

\begin{table}
\centering
\small
\setlength{\tabcolsep}{9pt}
\caption{\label{tbl:verification-time}Seconds to verify combinations of interactions and invariants of the two example applications. Each entry represents the mean verification time over 5 runs with the deviation shown in parentheses.}
\begin{tabularx}{.8\columnwidth}{ Xccc }
\toprule
\multicolumn{4}{c}{\textbf{Distributed Calendar}}\\
\midrule
\textbf{Interaction} & \multicolumn{3}{c}{\textbf{Invariant}} \\
\midrule
& 1 & 2&\\
\cmidrule{2-3}
Add vacation & 3.32 ($\pm$ 0.05) & 2.97 ($\pm$ 0.03)&\\
Remove vacation & 3.28 ($\pm$ 0.06) & 3.00 ($\pm$ 0.02)&\\
Change vacation & 3.32 ($\pm$ 0.05)& 3.04 ($\pm$ 0.03)&\\
Add work & 3.31  ($\pm$ 0.04) & -- & \\
Remove work & 3.30 ($\pm$ 0.06)& -- & \\
Change work & 3.34 ($\pm$ 0.06)& -- &\\
\midrule
\multicolumn{4}{c}{\textbf{TPC-C}}\\
\midrule
\textbf{Interaction} & \multicolumn{3}{c}{\shortstack{\textbf{Consistency} \textbf{Condition}}} \\
\midrule
& 3 & 5 & 7 \\
\cmidrule{2-4}
New Order & 45.4 ($\pm$ 63.69) & 7.63 ($\pm$ 0.11) & 14.49 ($\pm$ 7.31)\\
Delivery & 5.78 ($\pm$ 0.03) & 5.74 ($\pm$ 0.07) & 5.76 ($\pm$ 0.09)\\
\bottomrule
\end{tabularx}
\end{table}

\emph{Results.}
Verification times differed depending mainly on the complexity and length of the transactions and invariants under consideration.
Differences become apparent especially when looking at the results for TPC-C.
Proofs involving the \emph{New Order} interaction, which is the most "write-heavy" interaction of TPC-C that changes many source reactives at once, generally took longer to verify than others. For \emph{New Order}, we also observe a much higher deviation of up to 64 seconds which we assume to be caused by internal Z3 heuristics\footnote{
These could likely be improved by annotating quantifiers in invariants and pre-/postconditions with hand-crafted trigger expressions~\cite{MuellerSchwerhoffSummers16}.}.
When interpreting the results, it is important to note that
each interaction/invariant combination has to be verified only once and independently of other combinations.
Large-scale applications can be verified step-by-step by splitting them into smaller pieces.
Furthermore, we limit the need for verification to potential conflicts that we derive from the reactive data-flow graph.
Programmers can add new functionality to the application (i.e., specify interactions) and only have to reason about the properties of that new functionality (i.e., specify its invariants) and the system ensures global safety -- at only the cost of the amount of overlap with existing functionality.
This allows for an incremental development style, where only certain parts of programs have to be (re-)verified, when they have been
changed or added.

\section{Related Work}
\label{section-related-work}

Our work
relates to three areas:
distributed data types, formal reasoning, and language-based approaches.
Sections below relate work from each area to respective aspects of our approach.

\hypertarget{consistency-through-distributed-data-types}{%
\subsection{Consistency Through Distributed Data
Types}\label{consistency-through-distributed-data-types}}

\emph{Conflict-Free Replicated Data Types (CRDTs)}~\cite{Shapiro2011,preguica2018}
are a building block for
constructing systems and applications that guarantee eventual consistency.
CRDTs are used in distributed database systems such as
\emph{Riak}~\cite{Klophaus2010} and \emph{AntidoteDB}~\cite{Akkoorath2016}.
These databases make it possible to construct applications that behave under mixed consistency,
but unlike our approach, they leave reasoning about application guarantees to the programmer.
\authorcite[Gomes et al.]{Gomes2017} and \authorcite[Nair et al.]{Nair2020} propose frameworks for formally verifying the correctness of CRDTs
while \emph{VeriFx}~\cite{deporre2023} and \emph{Propel}~\cite{zakhour2023} are dedicated language-based solutions for specifying and verifying replicated data types.

Several works suggest distributed data types that extend the
capabilities of CRDTs:
\emph{Mergeable Replicated Data Types (MRDTs)}~\cite{Kaki:2019:MRD:3366395.3360580}
automatically synthesize merge functions from
relational specifications of the data type.
\emph{Katara}~\cite{laddad2022} synthesizes a verified CRDT from the specification of a sequential data type.
\authorcite[De Porre et al.]{DePorre2019} suggest \emph{strong
eventually consistent replicated objects (SECROs)} relying on a replication protocol that tries to find a valid total order of all operations.
Follow-up work on \emph{explicitly consistent replicated objects (ECROs)}\cite{DBLP:journals/pacmpl/PorreFPB21} reorders operations if possible and add coordination in cases where reordering is not sufficient.
Similarly, \emph{Hamsaz}~\cite{Houshmand2019} combines the specification of a sequential object with high-level invariants to synthesize a replicated object that satisfies the invariants.
All approaches above tie consistency and safety properties to specific data types/objects.
This is not sufficient to guarantee end-to-end correctness of an entire local-first application - consistency bugs can still manifest in derived information (e.g., in the user interface).

\hypertarget{formal-reasoning-about-consistency-levels}{%
\subsection{Automated Reasoning about Consistency
Levels}\label{formal-reasoning-about-consistency-levels}}

Our formalization
is in part inspired by the work of \authorcite[Balegas et al.]{Balegas2015, balegas2015a} on \emph{Indigo}.
The work introduces a database middleware consisting of transactions and invariants
to determine the ideal consistency level - called \emph{explicit consistency}.
They build on the notion of \emph{invariant-confluence} for
transactions that cannot harm an invariant which was first introduced by \authorcite[Bailis et al.]{bailis2014}.
While they work on a database level, we show how to integrate this reasoning
approach into a programming language (Section~\ref{sec:programming-model}).
An important difference between our \emph{invariant-confluence} and the one by
\authorcite[Balegas et al.]{Balegas2015} is that our approach also verifies local preservation of invariants, whereas their reasoning principle assumes invariants to always hold in a local context.
In a more recent work called \emph{IPA}, \authorcite[Balegas et al.]{Balegas2018} propose a
static analysis technique that aims at automatically repairing
transaction/invariant conflicts without adding synchronization between
devices.
We consider this latter work complementary to ours.
\authorcite[Whittaker and Hellerstein]{Whittaker2018} also build on the idea
of invariant-confluence and extend it to the concept of \emph{segmented
invariant-confluence}. Under segmented invariant-confluence, programs
are separated into segments that can operate without coordination and
coordination only happens in between the segments. The idea is similar
to our definition of \emph{conflicting interactions}, however,
their procedure cannot suggest a suitable program segmentation, but requires developers to supply them.

\authorcite[Gotsman et al.]{Gotsman2016} propose an alternative formalization and proof rule which allows reasoning about the invariant-safety of database transactions. Their formalization employs a token-based synchronization protocol similar to ours (see Section~\ref{sec:programming-model}).
The \emph{SIEVE} framework~\cite{Li2014} uses invariants and program annotations to infer where a Java program can safely
operate under CRDT-based replication and where strong consistency is
necessary. They do so by relying on a combination of static and dynamic
analysis techniques. Compared to \emph{SIEVE}, our formal reasoning
does not require any form of dynamic analysis.
\emph{Blazes}~\cite{Alvaro2014}
is another analysis framework that uses programmer supplied specifications to
determine where synchronization is necessary to ensure eventual consistency.
Unlike \emph{Blazes}, \langname{} ensures that programs
are "by design" at least eventually consistent,
while also allowing the expression and analysis of programs that need stronger consistency.
\emph{Q9}~\cite{Kaki2018} is a bounded
symbolic execution system, which identifies invariant violations caused
by weak consistency guarantees.
Similar to our work, \emph{Q9} can determine where exactly stronger consistency guarantees are needed to maintain certain application invariants.
However, its verification technique is bound by the number of possible concurrent operations.
\langname\ can provide guarantees for an unlimited amount
of devices with an unlimited amount of concurrent operations (see
Section~\ref{sec:programming-model}).

\hypertarget{language-abstractions-for-data-consistency}{%
\subsection{Language Abstractions for Data
Consistency}\label{language-abstractions-for-data-consistency}}

We categorize language-based approaches based on how they achieve consistency
and on the level of programmer involvement.

\emph{Manual Choice of Consistency Levels.}
Approaches in this category expose consistency levels as language abstractions.
\authorcite[Li et al.]{Li2012}
propose \emph{RedBlue Consistency} where programmers manually label
their operations to be either blue (eventually consistent) or red
(strongly consistent).
In MixT~\cite{Milano2018}, programmers annotate classes with different consistency levels and the system uses an information-flow type system to ensure that the requested guarantees are maintained. However, this still requires expert knowledge about each consistency level, and wrong choices can violate the intended program semantics.
Other approaches~\cite{Myter2018, DBLP:journals/pacmpl/KohlerEWMS20}
expect programmers to choose between \emph{consistency} and
\emph{availability}, again leaving the reasoning duty about consistency levels to the programmer.
Compared to \langname{}, languages in this category place higher burden on programmers: They decide which operation needs which consistency level, a non-trivial and error-prone selection.

\emph{Automatically Deriving Consistency from Application Invariants.}
Approaches in this category relieve programmers from reasoning about consistency levels
based on some form of programmer annotations about application invariants.
\emph{CAROL} \cite{Lewchenko2019} uses CRDTs to replicate data and features a refinement typing discipline
for expressing safety properties similar to our \emph{invariants}.
Carol makes use of pre-defined data types with \emph{consistency guards} used by the type system to check for invariant violations.
The compatibility of datatype operations and consistency guards is verified ahead of time using an algorithm for the Z3 SMT solver.
This approach hides much of the complexity from the programmer, but the abstraction breaks once functionality that is not covered by a pre-defined datatype is needed.
Unlike Carol, \langname{} does not rely on predefined consistency guards,
but allows the expression of safety properties as arbitrary logical formulae.
Additionally, \emph{CAROL} only checks the concurrent interactions of a program for invariant violations,
whereas \langname{} verifies the overall application including non-distributed parts.
\authorcite[Sivaramakrishnan et al.]{sivaramakrishnanDeclarativeProgrammingEventually2015} propose \emph{QUELEA}, a declarative language for programming on top of eventually consistent datastores.
It features a contract-language to express application-level invariants and automatically generates coordination strategies in cases where invariants could be violated by concurrent operations.
\emph{QUELEA}'s contract-language requires programmers to express the desired properties using low-level visibility relations,
which can be challenging to get right for non-experts.
\langname{} avoids this intermediate reasoning and automatically derives the right level of consistency for satisfying high-level safety invariants to enable end-to-end correctness.

\emph{Automating Consistency by Prescribing the Programming Model.}
Languages in this category
seek to
automate consistency decisions by
prescribing a certain programming model such that certain consistency problems are impossible to occur.
In \emph{Lasp}~\cite{Meiklejohn2015},
programmers model the data flow of their applications using combinator functions on CRDTs.
Programs written in \emph{Lasp} always provide eventual consistency but
contrary to \langname, \emph{Lasp} does not allow arbitrary compositions of distributed data types.
\emph{Bloom}~\cite{Alvaro2011}
provides programmers with ways to write programs that are \emph{logically monotonic}
and therefore offer automatic eventual consistency.
Both \emph{Lasp} and \emph{Bloom}, however, are not meant to formulate programs that need stronger consistency guarantees.
\langname{} is similar to \emph{Lasp} and \emph{Bloom}
in the sense that we also prescribe a specific -- reactive -- programming style.
However, our programming model is less restrictive and allows arbitrary compositions of distributed datatypes. This is enabled by leveraging the composability properties of reactive data-flow graphs. Secondly, \langname{} provides a principled way to express hybrid consistency applications with guarantees stronger than eventual consistency. %
\authorcite[Drechsler et al.]{Drechsler2018} and \authorcite[Mogk et al.]{Mogk2018,Mogk2019} also use a reactive programming model
similar to ours
to automate consistency in presence of multi-threading respectively of a distributed execution setting.
However,
they do
not support a hybrid consistency model. \authorcite[Drechsler et al.]{Drechsler2018} enable strong consistency (serializability)
only, while \authorcite[Mogk et al.]{Mogk2018,Mogk2019} support only eventual consistency.

\hypertarget{sec:chap:conclusion}{%
\section{Conclusion and Future Work}\label{sec:chap:conclusion}}

In this paper, we proposed \langname{}, a language for local-first software with verified safety guarantees. 
\textit{\langname} combines the declarative data flow of reactive programming with static analysis and
verification techniques 
to precisely determine concurrent interactions 
that could violate programmer-specified safety properties.
We presented a formal definition of the programming model
and a modular verification that detects concurrent executions that may violate application invariants. 
In case of invariant violation due to concurrent execution, 
\langname{} automatically enforces the necessary amount of coordination.
\langname{}'s verifying compiler  translates \langname{} programs to Viper~\cite{MuellerSchwerhoffSummers16} 
for automated verification and to Scala for the application logic including synthesized synchronization to guarantee the specified 
safety invariants. An evaluation of \langname{}'s programming model in two case studies confirms that it facilitates the
development of safe local-first applications and enables efficient and modular automated 
reasoning about an application's safety properties.
Our evaluation shows that verification times are acceptable
and that the verification effort required from developers is reasonable.

In the future, it would be desirable to integrate existing libraries of verified CRDTs~\cite{Gomes2017} or even solutions that allow ad-hoc verification of CRDT-like data types~\cite{Nair2020,laddad2022}.
This would enable us to support a wider range of data types or even allow programmers to use custom distributed data types, 
which can be verified to be eventually consistent.
Furthermore, our current data-flow analysis is limited to static data-flow graphs.
While static reasoning about dynamic graphs is impossible in the general case, 
most applications make systematic use of dynamic dependencies, 
and we believe it would be feasible to support common cases.

\bibliographystyle{ACM-Reference-Format}
\bibliography{bibliography}

\end{document}